\documentclass[12pt, a4paper]{article}

\usepackage{amsfonts}
\usepackage{amsmath}
\usepackage{amssymb}
\usepackage{amsthm}
\usepackage{mathtools}
\usepackage{authblk}
\usepackage{appendix}
\usepackage{bm}
\usepackage{bbm}
\usepackage{dsfont}
\usepackage{gensymb}
\usepackage{natbib}
\usepackage{subcaption}
\usepackage{lineno}
\usepackage{hyperref}
\usepackage{enumitem}
\usepackage[T1]{fontenc}
\usepackage{multirow}
\usepackage{soul}
\usepackage[font = small]{caption}

\usepackage{graphicx,url}

\newtheorem{remark}{Remark}
\newtheorem{assumption}{Assumption}
\newtheorem{proposition}{Proposition}
\newtheorem{definition}{Definition}

\theoremstyle{definition}
\newtheorem{example}{Example}
\usepackage[utf8]{inputenc}

\newcommand{\rd}{\mathrm{d}}
\newcommand{\rD}{\mathrm{D}}
\newcommand{\rH}{\mathrm{H}}

\newcommand{\rQ}{\mathrm{Q}}
\newcommand{\rX}{\mathrm{X}}

\newcommand{\rZ}{\mathrm{Z}}

\newcommand{\E}{\mathrm{E}}
\newcommand{\e}{\mathrm{e}}

\def\T{{ \mathrm{\scriptscriptstyle T} }}

\graphicspath{ {../images/} }

\newcommand{\blind}{1}
\newcommand{\spacing}{1.1}

\begin{document}

\def\spacingset#1{\renewcommand{\baselinestretch}%
{#1}\small\normalsize} \spacingset{1}

\if1\blind
{
  \title{\bf Bias correction of quadratic spectral estimators}
    \author[1]{Lachlan Astfalck}
    \author[2]{Adam Sykulski}
    \author[1]{Edward Cripps}
    \affil[1]{School of Physics, Mathematics \& Computing, The University of Western Australia, Australia}
    \affil[2]{Department of Mathematics, Imperial College London, U.K.}
    
    \setcounter{Maxaffil}{0}
    \renewcommand\Affilfont{\itshape\small}
  \maketitle
} \fi

\if0\blind
{
  \bigskip
  \bigskip
  \bigskip
  \begin{center}
    {\LARGE\bf Bias correction of quadratic spectral estimators}
\end{center}
  \medskip
} \fi

\bigskip
\begin{abstract}
  The three cardinal, statistically consistent, families of non-parametric estimators to the power spectral density of a time series are lag-window, multitaper and Welch estimators.  However, when estimating power spectral densities from a finite sample each can be subject to non-ignorable bias. \cite{astfalck2024debiasing} developed a method that offers significant bias reduction for finite samples for Welch's estimator, which this article extends to the larger family of quadratic estimators, thus offering similar theory for bias correction of lag-window and multitaper estimators as well as combinations thereof. Importantly, this theory may be used in conjunction with any and all tapers and lag-sequences designed for bias reduction, and so should be seen as an extension to valuable work in these fields, rather than a supplanting methodology. The order of computation is larger than $\mathcal{O}(n \log n)$ typical in spectral analyses, but not insurmountable in practice. Simulation studies support the theory with comparisons across variations of quadratic estimators.
\end{abstract}

\noindent%

\spacingset{\spacing}

\section{Quadratic Spectral Estimators} \label{sec:intro}

Denote by $\{\rX_t\}$ a stationary zero-mean real-valued stochastic process, discretely observed at interval $\Delta$ and indexed by $t \in \mathcal{Z}$. Without loss of generality, here, and in what follows, we assume $\Delta = 1$. Assume the power spectral density $f(\omega)$ exists, and so 
\begin{equation*} \label{eqn:fourier}
  \gamma(\tau) = \frac{1}{2\pi}\int_{-\pi}^{\pi} f(\omega) \e^{i \omega \tau} \; \mathrm{d} \omega, \hspace{3mm} f(\omega) = \sum_{\tau = -\infty}^{\infty} \gamma(\tau) \e^{-i \omega \tau}
\end{equation*}
where $\gamma(\tau) = \E[\rX_t \rX_{t-\tau}]$ is the auto-covariance sequence, $\tau \in \mathcal{Z}$, and  $\omega \in [-\pi, \pi]$ is defined in radians. We make the following assumption on $\{\rX_t\}$.
\begin{assumption}
 Let $\{\rX_t\}$ be a general linear process so that $\rX_t = \sum_{i = -\infty}^\infty \theta_i \epsilon_{t-i}$ with independent and identically distributed $\epsilon_t$ with $\E[\epsilon_t] = 0$, $\E[\epsilon_t^2] = 1$, $\E[\epsilon_t^4] < \infty$, and the regularity condition $\sum_{\tau = - \infty}^\infty (1 + |\tau|) |\gamma(\tau)|< \infty$, implying $f(\omega)$ has a bounded and uniformly continuous derivative.
\end{assumption}
The predominant task in spectral density estimation is to obtain some estimate of the power spectral density; herein, we concern ourselves with the case of non-parametric estimation of $f(\omega)$. The periodogram is the fundamental estimator to the power spectral density, but is statistically inconsistent. Broadly, consistency is resolved by one of three families of non-parametric estimators: lag-windows, multitapers and Welch estimators. Delineating all non-parametric estimators into one of these three methodologies is not always obvious; for instance, Bartlett's estimator may be seen both as a special case of Welch's estimator or as a lag-window estimator with a triangular lag-sequence. This literature is too vast to adequately review here, and we defer the reader to one of the many comprehensive book-length treatments such as \cite{percival2020spectral,priestley1981spectral,politis2019time} for more detailed information. We find these to be the categorisations and naming conventions most common in the present literature, but acknowledge that alternatives exist \citep[for example, see pp.252][]{percival2020spectral}. 

Quadratic spectral estimators describe a generalising family of non-parametric spectral estimators \citep{percival1993spectral,walden1995effective,walden2000unified}. Define the complex valued demodulation process $\rZ_t = \rX_t \e^{i \omega t}$, which has corresponding auto-covariance function, $\gamma_\rZ(\tau) = \gamma(\tau) \e^{i \omega \tau}$, and power spectral density $f_\rZ(\omega') = f(\omega - \omega')$, in particular, $f_\rZ(0) = f(\omega)$. Given an $n$-dimensional finite observation from $\{\rX_t\}$, define $\rX$ and $\rZ$ as the column vectors $\rX = (\rX_0, \dots, \rX_{n-1})$ and $\rZ = (\rZ_0, \dots, \rZ_{n-1})$. In what follows, denote by superscript $\rH$ the Hermitian transpose and by an asterisk the complex conjugate. 

\begin{definition}
  Define a quadratic spectral estimator $I_\mathrm{Quad}(\omega)$ as
  \begin{equation} \label{eqn:quadratic}
    I_\mathrm{Quad}(\omega) = \rZ^\rH \rQ \rZ = \sum_{s = 0}^{n-1} \sum_{t = 0}^{n-1} \rZ^*_s \rQ_{s,t} \rZ_t
  \end{equation}
  where $\rQ_{s,t}$ is the $(s,t)$th element of the real valued and symmetric $n \times n$ matrix $\rQ$.
\end{definition}

Choosing the values of $\rQ$ determines the properties of the estimator and, as we now demonstrate, encompasses the periodogram, lag-window, multitaper, and Welch estimators. 

\begin{example}[Periodogram] \label{exa:period}
The periodogram, $I_n(\omega)$, may be defined as the squared modulus of the discrete Fourier transform of $\rX$,
  \begin{equation*}
    I_n(\omega) = \left| \sum_{t = 0}^{n-1} h_t \rX_t \e^{-i \omega t} \right|^2 = \sum_{s = 0}^{n-1} \sum_{t = 0}^{n-1} h_s h_t \rX_s \rX_t \e^{-i \omega (s-t)} = \sum_{s = 0}^{n-1} \sum_{t = 0}^{n-1} \rZ_s^* h_s h_t \rZ_t = \rZ^\rH \rQ \rZ
  \end{equation*}
  where $h_t$ represents the chosen data-taper and $\sum_{t = 0}^{n-1} h_t^2 = 1$.
  The case of $h_t = n^{-1/2}$ corresponds to the choice of `no taper' for the raw periodogram. 
  The periodogram is therefore a quadratic spectral estimator with values $\rQ_{s,t} = h_s h_t$ and $\rQ$ is of unit rank.
\end{example}

\begin{example}[Lag-window] \label{exa:lag_window}
  Define a lag-window estimator as $I_\mathrm{Lag}(\omega) = \sum_{\tau = -n+1}^{n-1} (w_\tau \cdot \hat{\gamma})(\tau) e^{-i \omega \tau}$ where $\hat{\gamma}(\tau) = \sum_{t = 0}^{n - |\tau| - 1} h_{t+|\tau|} \rX_{t+|\tau|} h_t \rX_t$ and $\{w_\tau\}$ is the lag-sequence. We may re-write $I_{\mathrm{Lag}}(\omega)$ as
  \begin{equation*}
    I_{\mathrm{Lag}}(\omega) = \sum_{s = 0}^{n-1} \sum_{t = 0}^{n-1} w_{(s-t)} h_s h_t \rX_s \rX_t \e^{-i \omega (s-t)}
  \end{equation*}
  so that $I_{\mathrm{Lag}}(\omega)$ is a quadratic spectral estimator where $\rQ_{s,t} = w_{(s-t)} h_s h_t$ and $\rQ$ may be up to full rank.
\end{example}

\begin{example}[Welch's estimator] \label{exa:welch}
  Welch's estimator is defined $I_\mathrm{Welch}(\omega) = M^{-1} \sum_{m = 0}^{M-1} I_L^m(\omega)$ where $I_L^m(\omega)$ is the periodogram calculated from the $m$th length-$L$ segment of the data. We may therefore define the $I_{L}^{m}(\omega)$ as periodograms calculated from the full length-$n$ observation, so that 
  \begin{equation} \label{eqn:welch_period}
    I_{L}^{m}(\omega) = \left| \sum_{t=0}^{n-1} h_t^m \rX_{t} \e^{-i\omega t} \right|^2
  \end{equation}
  where, here, $\{h^m_t\}$ is a length $L$ data taper that has been pre- and post-pended with $mLp$ and $n-L(1+mp)$ zeros, respectively, and $p$ is the degree of segment overlap. From Example~\ref{exa:period}, we can write \eqref{eqn:welch_period} in quadratic form $I_{L}^{m}(\omega) = \rZ^\rH \rQ^m \rZ$ where each $\rQ^m$ is the outer product of $\{h^m_t\}$ with itself. Thus, $I_{\mathrm{Welch}}(\omega) = \rZ^\rH \rQ \rZ$ is a quadratic spectral estimator where $\rQ = \frac{1}{M}\sum_{m = 0}^{M-1} \rQ^m$ and is of rank-$M$. 
\end{example}

\begin{example}[Multitaper] \label{exa:multi}
  Multitaper estimators are defined as $I_\mathrm{Multi}(\omega) =  \sum_{k = 0}^{K-1} d_k I_n^k(\omega)$, where $\sum_{k=0}^{K-1} d_k = 1$ and the $I_n^k(\omega)$ are periodograms calculated with tapering sequence $\{h_{k,t}\}$, defined orthogonally over $k$ so that $\sum_{t=0}^{n-1}h_{k,t} h_{k',t} = 0$ for $k \neq k'$. Similarly to Example~\ref{exa:welch}, $I_{n}^{k}(\omega) = \rZ^\rH \rQ^k \rZ$ where each $\rQ^k$ is the outer product of $\{h_{k,t}\}$ with itself, and $I_\mathrm{Multi}(\omega) = \rZ^\rH \rQ \rZ$ is a quadratic spectral estimator where $\rQ = \sum_{k = 0}^{K-1} d_k \rQ^k$ and is of rank-$K$.
\end{example}

\section{Statistical Properties}

Assume $\rQ$ has $\mathrm{tr}(\rQ) = 1$, as in all previous examples, and is of rank $K$ for $1 \leq K \leq n$. If $\rQ$ is positive semi-definite then $I_\mathrm{Quad}(\omega)$ is guaranteed to be non-negative. It is common to assume positive semi-definite $\rQ$; although, as this excludes certain lag-window estimators \citep[e.g.][ described in Example~6]{politis1995bias} we relax this assumption. The spectral decomposition of $\rQ$ is $\rQ = \rH \rD \rH^\T = \sum_{k=0}^{K-1} d_k \rH_k \rH_k^\T$ where $\rH$ is an orthogonal $n \times K$ matrix with columns $\rH_0, \dots, \rH_{K-1}$, and $\rD$ is a $K \times K$ diagonal matrix with values $d_0, \dots, d_{K-1}$ such that $\sum_{k=0}^{K-1} d_k = 1$ because $\mathrm{tr}(\rQ) = 1$. The $\rH_0, \dots, \rH_{K-1}$ and $d_0, \dots, d_{K-1}$ are the eigenvectors and eigenvalues of $\rQ$, respectively. By definition, the eigenvectors $\rH_k$ respect unit square-summability and so are valid tapers. Thus, any quadratic spectral estimator \eqref{eqn:quadratic} can be written as
\begin{equation*}
  I_\mathrm{Quad}(\omega) = \sum_{k=0}^{K-1} d_k \left| \sum_{t=0}^{n-1} h_{k,t} \rX_t \e^{-i \omega t} \right|^2 = \sum_{k=0}^{K-1} d_k I^k_n(\omega)
\end{equation*}
which is identical to the form of the multitaper, but with the tapers defined by the eigenvectors of $\rQ$. Defining $\mathcal{H}_n^k(\omega) = \left| \sum_{t=0}^{n-1} h_{k,t} \e^{-i \omega t} \right|^2$, the expectation of each $I^k_n(\omega)$ is $\E[I^k_n(\omega)] = (f \ast \mathcal{H}_n^k)(\omega)$. Therefore,
\begin{equation} \label{eqn:quad_exp}
  \E[I_\mathrm{Quad}(\omega)] = \sum_{k=0}^{K-1} d_k \E[I^k_n(\omega)] = \left(f \ast \bar{\mathcal{H}}_n\right)(\omega)
\end{equation}
where $\bar{\mathcal{H}}_n(\omega) = \sum_{k=0}^{K-1} d_k \mathcal{H}_n^k(\omega)$. For all $f(\omega)$ aside from white noise, $I_\mathrm{Quad}(\omega)$ is biased for finite $n$. As seen in Examples~\ref{exa:period}--\ref{exa:multi}, $\bar{\mathcal{H}}_n(\omega)$ and thus 
\begin{equation} \label{eqn:bias_def}
  \mathrm{bias}[I_\mathrm{Quad}(\omega)] = \E[I_\mathrm{Quad}(\omega)] - f(\omega) = \left(f \ast \bar{\mathcal{H}}_n\right)(\omega) - f(\omega)
\end{equation}
are dependent on the appointed tapers and lag-windows. The variance of $I_\mathrm{Quad}(\omega)$ is
\begin{equation}\label{eqn:quad_variance}
  \mathrm{var}[I_\mathrm{Quad}(\omega)] = \sum_{k = 0}^{K-1} d_k^2 \mathrm{var}[I_n^k(\omega)] + 2\sum_{j < k} d_j d_k\mathrm{cov}[I_n^j(\omega), I_n^k(\omega)] = \sum_{k=0}^{K-1}d_k^2 \mathrm{var}[I_n^k(\omega)] 
\end{equation}
where the second equality results from $\mathrm{cov}[I_n^j(\omega), I_n^k(\omega)] = 0$ due to the orthogonality across the $\rH_k$. The covariance $\mathrm{cov}[I_n^k(\omega), I_n^k(\omega+\eta)]$, from which $\mathrm{var}[I_n^k(\omega)]$ in \eqref{eqn:quad_variance} is obtained, is calculated as
\begin{equation} \label{eqn:covariance}
\begin{split}
\mathrm{cov}[I_n^k(\omega), I_n^k(\omega+\eta)] =& \left| \int_{-\pi}^{\pi} H^*_k(\omega + \eta - \omega')H_k(\omega-\omega') f(\omega')\; \rd \omega'\right|^2 + \\ 
&\quad \quad \quad \underbrace{\left| \int_{-\pi}^{\pi} H_k(\omega + \eta + \omega')H_k(\omega - \omega') f(\omega') \; \rd \omega'\right|^2}_{= \mathcal{O}(\log^2 n_0/n_0^{2})}  + \mathcal{O}(n_0^{-1})  
\end{split}
\end{equation}
where $H_k(\omega) = \sum_{t=0}^{n-1} h_{k,t} \e^{-i \omega t}$ and $n_0$ is the number of non-zero elements in $\{h_{k,t}\}$. This is an extension of Theorem~5.2.5 of \cite{brillinger2001time} to the case of general taper $\{h_{k,t}\}$ obtained by substituting Equation~(211b) of \cite{percival2020spectral}. The $\mathcal{O}(\log^2 n_0/n_0^{2})$ term is obtained in \cite{krogstad1982covariance} for $h_{k,t} = n^{-1/2}$, and extends to the case of a general taper $\{h_{k,t}\}$ using Lemma~1 of \cite{astfalck2024debiasing}. 

Assuming $\mathrm{var}[I_n^k(\omega)]$ is constant over $k$, $\mathrm{M}^{-1} = \sum_{k=0}^{K-1}d_k^2$ is an approximation of the variance reduction in $I_\mathrm{Quad}(\omega)$, as compared to the periodogram, and $I_\mathrm{Quad}(\omega)$ is an $\sqrt{\mathrm{M}}$-convergent estimator. When $\rQ$ is positive semi-definite, $\mathrm{M}^{-1} \leq 1$. In actuality, $\mathrm{var}[I_n^k(\omega)]$ is not precisely constant over $k$, and so $\mathrm{M}^{-1}$ is not the exact factor of variance reduction; however, it is useful in establishing certain properties in Remark~\ref{rem:alpha}, below. To arrive at a notion of bias-variance trade-off, we require a metric to describe the effective number of independent observations from the spectral estimator, that is, the effective sample size. We achieve this via the bandwidth properties of $I_\mathrm{Quad}(\omega)$. Define the effective sample size as 
\begin{equation} \label{eqn:zeta}
    \zeta = \frac{2 \pi}{\int_{-\pi}^\pi \mathrm{cor}[I_z(\omega), I_z(\omega+\eta)] \; \mathrm{d}\eta}
\end{equation}
where $I_z(\omega)$ is the spectral estimator $I_\mathrm{Quad}(\omega)$ of a white noise process with constant $f(\omega')$, and so $\zeta$ is independent of $\omega$. Dividing \eqref{eqn:covariance} through by $f(\omega')$ and combining over the $K$ components yields $\mathrm{cor}[I_z(\omega), I_z(\omega+\eta)]$ which 
is calculated
\begin{equation} \label{eqn:R}
R(\eta) = \mathrm{M} \sum_{j,k=0}^{K-1}d_j d_k \left| \sum_{t=0}^{n-1} h_{j,t} h_{k,t} e^{-i \eta t}\right|^2 = \mathrm{cor}[I_z(\omega), I_z(\omega+\eta)] + \mathcal{O}(n_0^{-1}).
\end{equation}
Substituting \eqref{eqn:R} into \eqref{eqn:zeta} yields $\zeta = 1 / \mathcal{B}_\mathrm{Quad} + \mathcal{O}(n_0^{-1})$, where
\begin{equation} \label{eqn:r0}
\mathcal{B}_\mathrm{Quad} = \frac{1}{2\pi}\int_{-\pi}^{\pi} R(\eta) \; \rd\eta = \mathrm{M} \sum_{j,k=0}^{K-1} \sum_{t = 0}^{n-1}  d_j d_k h_{j,t}^2 h_{k,t}^2.
\end{equation}
Equation \eqref{eqn:r0} can be understood as the spectral bandwidth of $I_\mathrm{Quad}(\omega)$, i.e., the integral length scale of the spectral estimator, and is independent of $\omega$. The effective sample size $\zeta$ admits values $1 \leq \zeta \leq n$.

\begin{proposition} \label{prop:bias}
  Given $\{\rX_t\}$ satisfies Assumption~1, and $\omega \neq 0, \pm \pi$, the bias of a quadratic spectral estimator, $I_\mathrm{Quad}(\omega)$ is of order
  \begin{equation} \label{eqn:ess_order}
    \mathrm{bias}[I_\mathrm{Quad}(\omega)] = \mathrm{E}[I_\mathrm{Quad}(\omega)] - f(\omega) = \mathcal{O}\left(\frac{\log \zeta}{\zeta}\right).
  \end{equation}
\end{proposition}
\begin{proof}
  Re-write \eqref{eqn:bias_def} as $\mathrm{bias}[I_\mathrm{Quad}(\omega)] = \sum_{k = 0}^{K-1} d_k \{\E[I_n^k(\omega)] - f(\omega)\}$, and so $\mathcal{O}(\mathrm{bias}[I_\mathrm{Quad}(\omega)])$ is dominated by $\mathrm{max}_k[\{\mathcal{O}\{\E[I_n^k(\omega)] - f(\omega)\}]$. The maximum value for $\E[I_n^k(\omega)] - f(\omega)$ is achieved when the number of non-zero elements of $\{h_{k,t}\}$ are minimised. Equation \eqref{eqn:r0} implies a minimum number of non-zero elements in each $\{h_{k,t}\}$ that scales with $\zeta$ as $n$ grows.
  Following Lemma~1 of \cite{astfalck2024debiasing}, a taper with a minimum of $\zeta$ non-zero elements is bounded by $c_h \zeta^{-1/2}$ for some constant $c_h$, and so \eqref{eqn:ess_order} is established following the result of \cite{fejer1910lebesguessche}.
\end{proof}

\begin{remark} \label{rem:alpha}
  To satisfy $\sum_{t = 0}^{n-1} h_{k,t}^2 = 1$ with increasing $n$ for each $\{h_{k,t}\}$, all tapers scale at maximum $h_{k,t} \propto \zeta^{-1/2}$: see the proof of Proposition~\ref{prop:bias}. Consequently, $\zeta \mathrm{M} \propto n$ for a quadratic estimator, $I_\mathrm{Quad}(\omega)$. By apportioning $\zeta = \mathcal{O}(n^\alpha)$ and $\mathrm{M} = \mathcal{O}(n^{1-\alpha})$ the expected bias-variance trade-off is obtained. As $\mathrm{var}[I_\mathrm{Quad}(\omega)] = \mathcal{O}(\mathrm{M}^{-1})$, and from Proposition~1, the optimal value of $\alpha$ with respect to mean-squared-error is $\alpha \approx 1/3$. Further regularity assumptions on $f(\omega)$, or parameterisations of the $h_{k,t}$, may reduce the bias in \eqref{eqn:ess_order} and further lower the optimal value of $\alpha$, thus allowing for more variance reduction. We show two instances of this in Examples~\ref{exa:infinite_lag} and \ref{exa:sinusoidal}.
\end{remark}

\begin{example}[Parameterisation of Bartlett's estimator]
Bartlett's estimator arises as a special case of Welch's estimator and, as in Example~\ref{exa:welch}, is equivalent to a quadratic estimator with $d_k = 1/K$ and $h_{k,t}$ is a length-$L$ taper with values $L^{-1/2}$ pre- and post-pended with $kL$ and $n - L(k+1)$ zeros. Substituting these values in \eqref{eqn:r0} yields $\zeta = L$ and so $ML = n$ with equality, leading to the optimal choice of $\alpha \approx 1/3$.
\end{example}

\begin{example}[Parameterisation of an infinite-order flat-top lag-sequence] \label{exa:infinite_lag}
  \cite{politis1995bias,politis1999multivariate} specify the infinite-order flat-top lag-sequence for univariate $\{\mathrm{X}_t\}$, with
  \begin{equation*}
    w(\tau) = \begin{cases} 1 & |\tau| \leq 1/a \\ g(\tau; a) & 1/a < |\tau| \leq c/a \\ 0 & |\tau| > c/a \end{cases}
  \end{equation*}
  where $g(\tau; a)$ is a continuous real-valued function, defined over $\tau$ and parameterised by $a$, that satisfies $g(\tau; a) = g(-\tau; a)$, $g(\tau; 1) = g(\tau/a; a)$, $g(1/a; a) = 1$ and $g(c/a; a) = 0$. If $f(\omega)$ has $r$ bounded and continuous derivatives, then $\mathrm{bias}[I_\mathrm{Quad}(\omega)] = \mathcal{O}(\zeta^{-r})$. Calculated similar to Remark~\ref{rem:alpha} it is asymptotically optimal to choose $\alpha = 1/(2r + 1)$, where Assumption~1 states $r \geq 1$.
\end{example}

\begin{example}[Parameterisation of a sinusoidal multitaper] \label{exa:sinusoidal}
Sinusoidal multitapers are discrete approximations to minimum bias multitapers and are specified by
  \begin{equation*}
    h_{k, t} = \left(\frac{2}{n+1}\right)^{1/2} \sin \left\{ \frac{(k+1)\pi t}{n+1}\right\}
  \end{equation*}
  with weights $d_k = 1/K$ for $k \in \{0, \dots, K-1\}$ \citep{riedel1995minimum}. This leads to $\mathcal{B}[I_\mathrm{Quad}(w)] = (K+1)/(n+1)$ \citep{walden1995effective}, thus $\zeta \propto n/K$, $\mathrm{M} = K$ and $\zeta \mathrm{M} \propto n$ as desired. As shown in \cite{riedel1995minimum}, the bias of a Sinusoidal multitaper is $\mathcal{O}(\zeta^{-2})$, and therefore suggests an optimal choice of $\alpha = 1/5$.
\end{example}

\begin{remark}
Studying the general theory of quadratic estimators provides a principled framework with which to combine classes of non-parametric spectral estimators that are typically siloed. For instance, multitaper estimators based on Slepian sequences with large $K$ are known to re-introduce bias as $K$ increases; but without increasing $K$ we cannot resolve variance. Often in scientific application there are requirements for variance to be resolved under some threshold and may require a choice of $\alpha$ that does not necessarily adhere to asymptotic arguments. For example, Welch and multitaper methods can be combined by segmenting the time-series, calculating multitaper estimates on each segment for some fixed $K$, and then averaging the segments to resolve variance; see \cite{schmidt2022spectral} for an example. Such an estimator remains a quadratic estimator, and so we may obtain the statistical properties as described herein. Historically, quadratic estimators were predominantly discussed to motivate the use of constructed multitaper estimators whereby the $\{h_{k,t}\}$ are specified directly rather than as a consequence of specifying $\rQ$. However, as discussed above, this need not be the only case. 
\end{remark}

\section{Bias correction of quadratic spectral estimators}

\subsection{Bias correction}

Assume a family of bases, such that the power spectral density may be represented as $f(\omega; \vartheta) = \sum_{s = 0}^{S-1} a_s b_s(\omega)$ where $b_s(\omega)$ is the $s$th basis function with corresponding auto-correlation function $\rho_s(\tau)$ obtained via an inverse Fourier transform, and $\vartheta = (a_0, \dots, a_{S-1})$. \cite{astfalck2024debiasing} proposed a methodology to debias Welch's estimator to $f(\omega)$ whereby the bases are intentionally biased, basis coefficients are inferred by fitting the biased bases to the estimate of $f(\omega)$ to obtain the $\hat{a}_s$ and a bias corrected estimate of $f(\omega)$ is obtained from $\sum_{s = 0}^{S-1} \hat{a}_s b_s(\omega)$. We now present the theory to extend these ideas to the general class of quadratic spectral estimators.

Define the Fourier frequencies $(\omega_0, \dots, \omega_{n-1}) = 2 \pi n^{-1} (-\lfloor n/2 \rfloor, \dots, \lceil n/2 \rceil-1)$, the vector $\mathrm{b}_s = \{b_s(\omega_0), \dots, b_s(\omega_{n-1})\}$ containing the evaluations of the $s$th basis over the Fourier frequencies, and the $S \times n$ matrix $\mathrm{B} = (\mathrm{b}_0^\T, \dots, \mathrm{b}_{S-1}^\T)$ that stores these evaluations over all $S$ bases. We bias the bases so that, $\check{b}_s(\omega) = (b_s \ast \bar{\mathcal{H}}_n)(\omega)$, where $\bar{\mathcal{H}}_n(\omega)$ defines the bias in $I_{\mathrm{Quad}}(\omega)$, see \eqref{eqn:quad_exp}. We calculate $\check{\mathrm{b}}_s = \{\check{b}_s(\omega_0), \dots, \check{b}_s(\omega_{n-1})\}$ via
\begin{equation} \label{eqn:bias_basis}
  \check{b}_s(\omega) = (b_s \ast \bar{\mathcal{H}}_n)(\omega) = 2 \times \mathrm{Re} \left[ \sum_{k = 0}^{K-1} d_k \left\{  \sum_{\tau = 0}^{n-1} \left(\sum_{t = 0}^{n - \tau - 1} h_{k,t} h_{k,t+\tau} \right) \rho_r(\tau) e^{-i \omega \tau} \right\}\right] - 1
\end{equation}
and further define the $S \times n$ matrix of biased bases evaluations as $\check{\mathrm{B}} = (\check{\mathrm{b}}_0^\T, \dots, \check{\mathrm{b}}_{S-1}^\T)$.

We model $I_{\mathrm{Quad}}(\omega) = \sum_{s=0}^{S-1} a_s \check{b}_s(\omega) + \epsilon(\omega)$, where $\mathrm{var}[\epsilon(\omega)] = \mathrm{var}[I_{\mathrm{Quad}}(\omega)]$ and for sufficiently large $\mathrm{M}$, the central limit theorem implies that $\epsilon(\omega)$ may be considered Gaussian. We expect there to be non-ignorable correlations between the frequencies for certain choices of $\mathrm{Q}$, which as discussed, is well represented by $R(\eta)$ in \eqref{eqn:R}. Different to \cite{astfalck2024debiasing}, we specify a circulant correlation matrix $\mathrm{W}$ with $(i,j)$th values $\mathrm{W}_{i,j} = \mathrm{cor}[I_\mathrm{Quad}(\omega_i), I_\mathrm{Quad}(\omega_j)] \approx R(\omega_i - \omega_j)$ and solve the weighted least squares problem
\begin{equation} \label{eqn:wls}
  \hat{\vartheta} = \underset{\vartheta}{\arg\min} \left\{ (\mathrm{I}_\mathrm{Quad} - \check{\mathrm{B}}^\T \vartheta)^\T \mathrm{V}^{-1} (\mathrm{I}_\mathrm{Quad} - \check{\mathrm{B}}^\T \vartheta) \right\}
\end{equation}
where $\mathrm{I}_\mathrm{Quad} = \{I_\mathrm{Quad}(\omega_0), \dots, I_\mathrm{Quad}(\omega_{n-1})\}$, $\mathrm{V} = \Gamma \mathrm{W} \Gamma$, and $\Gamma$ is a diagonal matrix with ($i,i$)th values $\Gamma_{i,i} = \mathrm{sd}[I_\mathrm{Quad}(\omega_i)]$ which, following Theorem~2 of \cite{astfalck2024debiasing}, is approximated by $I_\mathrm{Quad}(\omega_i)$. The basis coefficients $\hat{\vartheta}$ in \eqref{eqn:wls} are analytically available via the standard weighted least squares solution.

\subsection{Computation}

The computational order of the bias correction of the quadratic estimator is $\mathcal{O}\{\mathrm{max}(n S^2$, $nKS \log n)\}$, where $S$ is the number of bases and $K$ is the rank of $\rQ$. The first term is the computation required to calculate the solution to \eqref{eqn:wls} and the second term is the computation required to calculate the $\check{b}_s(\omega)$ in \eqref{eqn:bias_basis}. To ensure numerical stability, $S$ should scale at maximum with the effective sample size $\zeta$ so that $S = \mathcal{O}(n^\alpha)$, where $\alpha$ has a maximum value of $\alpha = 1/3$ (see Remark~\ref{rem:alpha}) and so $\mathcal{O}(nS^2) = \mathcal{O}(n^{5/3})$. The matrix $\mathrm{V}$ is circulant and so the inverse in \eqref{eqn:wls} is calculated with order $\mathcal{O}(n \log n)$ and does not affect the computational order of \eqref{eqn:wls}. For the computation in \eqref{eqn:bias_basis}, if $K \propto M = \mathcal{O}(n^{2/3})$, as in Welch's estimator and certain multitapers, then $\mathcal{O}(nKS \log n) = \mathcal{O}(n^2 \log n)$; if $K = n$ as in lag-windows,  $\mathcal{O}(nKS \log n) = \mathcal{O}(n^{7/3} \log n)$. At first, these rates may seem large, but they are a result of generalising the theory so as to establish the statistical properties. Instead, \eqref{eqn:bias_basis} may be computed faster for Welch and lag-window estimates: Welch estimates follow \cite{astfalck2024debiasing} and have order $\mathcal{O}(n \log n)$; lag-windows are computed similar to Equation~(9) of \cite{sykulski2019debiased} with the addition of a lag-sequence and have order $\mathcal{O}(n^{5/3}\log n)$. As opposed to \cite{astfalck2024debiasing} who retain $\mathcal{O}(n \log n)$ common to spectral analyses, we find a worse-case scenario rate of $\mathcal{O}(n^2 \log n)$, corresponding to a multitaper bias correction. Although this is much larger than $\mathcal{O}(n \log n)$, computation of an $\mathcal{O}(n^2 \log n)$ estimator is still possible for $n = \mathcal{O}(10^5)$ on a modern laptop computer. 

\subsection{Flexible basis functions} \label{sec:flex_bases}



In \cite{astfalck2024debiasing} the $b_s(\omega)$ were defined as rectangular bases, in effect representing a Riemannian approximation to $f(\omega)$. Assumptions on the order of differentiability of $f(\omega)$ can inform the bias that results from the Riemannian approximation: $\mathcal{O}(S^{-1})$ if $f(\omega)$ is of bounded first derivative, and $\mathcal{O}(S^{-2})$ if $f(\omega)$ is of bounded variation. The standard Riemannian basis is a polynomial of order 0 and the integrated mean squared errors can be improved upon by employing basis functions with increasing polynomial order and thus increasing degrees of differentiability. For instance, the trapezoidal rule has integrated mean square error $\mathcal{O}(S^{-3})$ and corresponds to a polynomial basis of order 1, and Simpson's rule has integrated mean square error $\mathcal{O}(S^{-4})$ and corresponds to a polynomial basis of order 2.

For the $s$th basis, $b_s(\omega)$, define $\omega^\mathrm{c}_s$ as the basis centre, $\delta_s$ as the basis width, and assume $b_s(\omega)$ to be $p$-times differentiable. We define $\rho_s(\tau)$, the auto-correlation function that corresponds to $b_s(\omega)$, as
\begin{equation}
  \rho_s(\tau; \omega^\mathrm{c}_s, \delta_s) = \left\{\delta_s^{-\frac{1}{p+1}} \mathrm{sinc}(\delta_s \tau) \exp(i \omega^\mathrm{c}_s \tau)\right\}^{p+1}
\end{equation}
where $\mathrm{sinc}(\cdot)$ is the normalised sinc function. For $p=0$, basis $b_s(\omega)$ is $b_{s}^0(\omega) = \mathrm{rect}(\omega/\delta_c-\omega_c)$, where $\mathrm{rect}(\cdot)$ is the rectangular function. Higher orders are obtained via the recursion
\begin{equation}
  b_{s}^p(\omega) = (b_{s}^{p-1} \ast b_{s}^0)(\omega)
\end{equation}
where the $b_{0}^p(\omega), \dots, b_{S}^p(\omega)$ define a
family of $p$-times differentiable functions, analogous to $p$-order B-splines defined over a circular domain. Choosing large $p$ imposes an assumption on the regularity conditions of $f(\omega)$; if $p$ is chosen correctly then an optimal rate of convergence, with respect to the basis approximation, may be guaranteed. However, if $p$ is overestimated then this may impose an assumption of smoothness on the underlying function that is inappropriate and so may re-introduce bias. In practice when $f(\omega)$ is unknown, $p$ must be estimated; see \cite{politis1995bias}, who demonstrate an empirical methodology to estimate $p$ from the empirical auto-covariance sequence $\hat{\gamma}(\tau)$.

\section{Simulation study with canonical {\sc AR(4)} model} \label{sec:sims}

\begin{figure} [b!]
  \centering
  \includegraphics[width = \linewidth]{"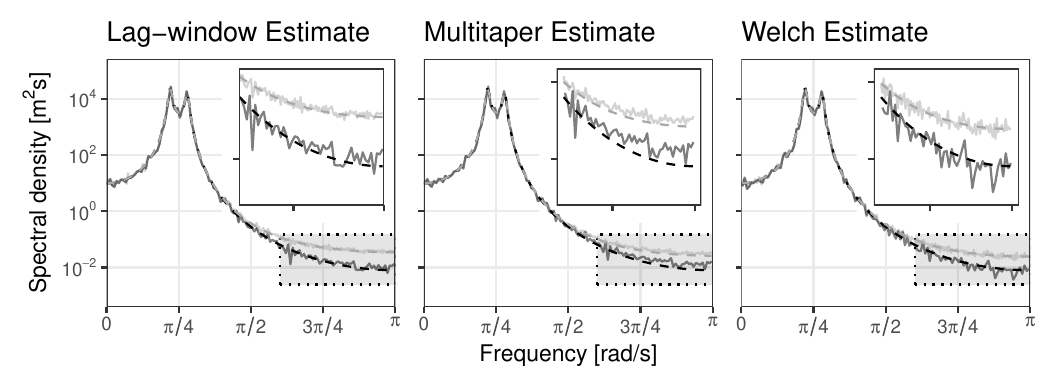"}
  \caption{Standard and debiased lag-window, multitaper and Welch estimates of a random {\sc{ar}(4)} process. The standard and debiased estimates are shown by the solid grey and black lines, respectively; the true process spectrum is shown by the black dashed line and the expectation of the standard estimator by the grey dashed line. Plot insets correspond to the shaded dotted boxes. Parameterisations of the {\sc{ar}(4)} process and the estimators are given in the main text.}
  \label{fig:single}
\end{figure}

We demonstrate our general bias correction technique on lag-window, multitaper and Welch estimates of the canonical {\sc AR(4)} model of \cite{percival1993spectral}, defined as $\rX_t = \sum_{p = 1}^4 \phi_p \rX_{t - p} + \epsilon_t$ where $\epsilon_t \sim \mathcal{N}(0, \sigma^2)$ and $\{\phi_1, \phi_1, \phi_1, \phi_1, \sigma\} = \{2.7607, -3.8106, 2.6535$, $-0.9238, 1\}$. We first present a bias correction to a single sample of length $n = 2^{14}$ for each of the three estimators, shown in Figure~\ref{fig:single}. We set $p=1$ for the bases, and parameterise each of the three estimators so as to give similar expected bias: the lag-window estimator uses a modified Daniell of width $M = 2^5$; the multitaper estimator specifies the tapers as $K = 2^5$ Slepian sequences with a time-bandwidth parameter of $2^4$; and the Welch estimator sets $L = 2^9$, $M = 2^5$, with zero-overlap and employs a Hamming data-taper on each segment. Estimates are plotted in Figure~\ref{fig:single} by the grey solid lines with their respective expectations, calculated as per \eqref{eqn:quad_exp}, plotted with the grey dashed lines. The bias corrected estimates are all given by the black solid lines, and the true $f(\omega)$ of the {\sc AR(4)} model is given by the black dashed line. The computational time to bias correct the lag-window and multitaper estimates was $\sim200\;\mathrm{ms}$, and $\sim5\;\mathrm{ms}$ for the Welch estimate. In all cases, the debiased estimator provides a better estimate of the true spectrum across all frequencies. Finally, in Figure~\ref{fig:ensemble} we show performance of the bias corrected estimators over a 1000-member ensemble of simulations and as a function of $\mathrm{M}$, the factor of variance reduction, with fixed $n/\mathrm{M}=2^{9}$ for each simulation. We observe a clear reduction of bias and root-mean-squared-error, across all estimators, and for all values of $\mathrm{M}$. An \texttt{R} implementation of the code is available at \url{github.com/astfalckl/dquad}.

\begin{figure}[h!]
  \centering
  \includegraphics[width = 0.95\linewidth]{"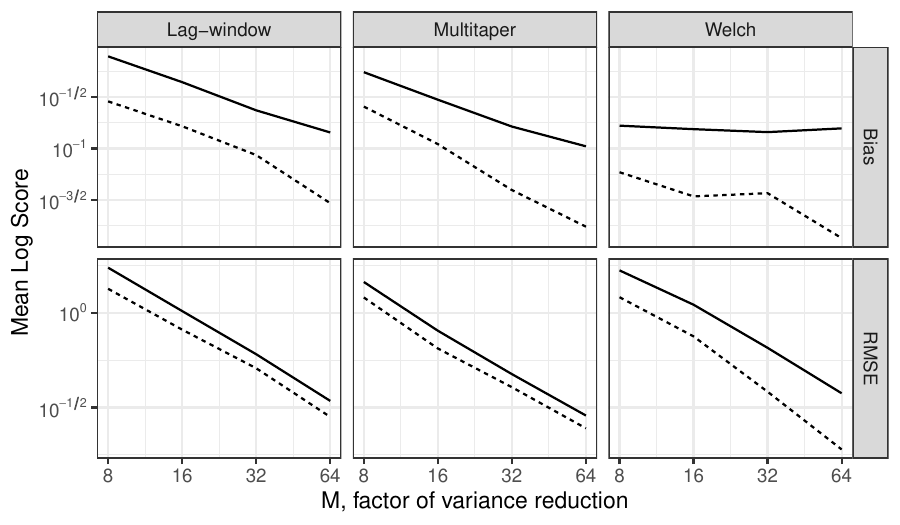"}
  \caption{Estimates of bias and root-mean-squared-error of the standard (solid) and debiased (dashed) lag-window, multitaper and Welch estimates as a function of the factor of variance reduction, $\mathrm{M}$. Each graphed point is obtained by calculating the metric from a 1000 member ensemble, for each frequency, and aggregating  over frequency by the mean log value. Parameterisations of the {\sc{ar}(4)} process and the estimators is given in the main text.}
  \label{fig:ensemble}
\end{figure}


\section*{Acknowledgement}
All authors are supported by the ARC ITRH for Transforming energy Infrastructure through Digital Engineering (TIDE), Grant No. IH200100009.

\spacingset{1}

\bibliographystyle{biometrika}
\bibliography{paper-ref}
\end{document}